\documentclass[10pt]{article}
\usepackage[english]{babel}
\usepackage[ansinew]{inputenc}
\usepackage[T1]{fontenc}

\usepackage{amsmath, amssymb}
\usepackage{amsthm}

\usepackage{siunitx}

\usepackage[backend=bibtex8]{biblatex}
\usepackage[babel]{csquotes}
\usepackage{framed}
\bibliography{D:/Promotion/02_Bibliothek/04_Literaturverzeichnis/Literatur}

\usepackage{csquotes}
\usepackage{pdfpages}

\usepackage{listings}

\newtheorem{prop}{Property}

\usepackage{pgfplots,xstring}
\pgfplotsset{compat=newest}
\usepackage{times}
\usepackage[author={Tobias Sehnke}]{pdfcomment}
\usepackage{authblk}
\usepackage[top=0.75in, bottom=1.0in, left=0.75in, right=0.75in]{geometry}
\title{Temporal Properties in Component-Based Cyber-Physical Systems \\ Appendix}
\author[1]{Tobias Sehnke}
\author[1]{Matthias Schultalbers}
\author[2]{Rolf Ernst}

\affil[1]{Powertrain Mechatronics, IAV GmbH, \{tobias.sehnke | matthias.schultalbers\}@iav.de}
\affil[2]{Institute of Computer and Network Engineering, Technische Universität Braunschweig, ernst@ida.ing.tu-bs.de}

\date{}

\begin{document}
\maketitle

In this document, we provide supplementary material to \cite{Sehnke2018}, which includes a more detailed description of the requirement transformations, outlined in Section 4.2 of the paper. For this purpose, we also provide a formal description of the temporal semantics model. 

\section{The Temporal Semantics Model}

\subsection{Events and Signals}

The temporal semantics model represents software as a composed set of components $C=\left\{c_1,c_2,\ldots\right\}$. An implementation of a component consists of a set of behaviors and ports, also called interfaces. Each behavior assigns values to output ports $\underline{\mathbf{y}}$ based on inner states and the values on the input ports $\underline{\mathbf{u}}$. The sampling ports $\underline{\mathbf{s}}$ and actuation ports $\underline{\mathbf{z}}$ provide a link to the physical environment. Throughout this document we use $\underline{\mathbf{x}}$ as a placeholder for any port, while $\mathbf{\underline{x}_{(i,j)}}$ addresses the $j^{th}$ interface of the component $c_i$. Each component consists of one or more executable units called runnables, which are assigned to schedulable units called a tasks $\tau$. 

Each occurrences at an interface $\underline{\mathbf{x}}$ is described by a data event $x^k \forall k\in\{1,\ldots,n \}$. An event is defined as the triple $x^k = \left (v_x^k, \hat{t}_x^k, t^k_x \right)$ whereas $v^k_x$ is the \textit{value}, $\hat{t}^k_x$ the \textit{tag} or timestamp and $t^k_x$ the so called \textit{logical timestamp}. The logical timestamp describes the temporal context of the physical state that is represented by the value. 

The ordered set of events that occur at the interface $\mathbf{\underline{x}}$ is called a signal $x = \left(x^1,\ldots,x^n \right) $. To each signal a set of signal paths $e_{x}=\{e_{x}^1,\ldots,e_{x}^n\}$ can be attributed, which describe the information flow to the corresponding interface of a signal. More specifically, a signal path $e_{x}^m \forall m\in\{1,\ldots,n \}$ is an ordered tuple, whose elements can be read, write, sampling or actuator interfaces. The causal relation of events is called the causal chain. To describe this more specifically for a given signal $x$ with a signal path $e^m_x$, which connects a sampling interface $\mathbf{\underline{s}}$ to the port of the signal $\mathbf{\underline{x}}$. Then, the causal chain $P^{(m,k,i)}_x$ describes a set of events $P^{(m,k,i)}_x = \left(s^r,\ldots,x^k\right)$, which are causally related to the event $x^k$. This set includes exactly one event for each interface in $e^m_x$. The set of all causal chains that can be assigned to an event is described by $P_x$. 

If a component changes the temporal context of information, we call this behavior algorithmic delay. The sum of all algorithmic delays in a signal paths is annotated as $d_x^k \forall k\in\{1,\ldots,n \}$, whereas we generally assume that $d_x$ is constant. Given a causal chain $P^{(m,k,i)}_x = \left(s^r,\ldots,x^k\right)$, which relates an event $x^k$ to a sampling event $s^r$, we can compute the logical timestamp formally as the sum of the tag $\hat{t}_s^r$ and the sum of algorithmic delays $d_x^k$
\begin{equation}
t_x^k = \hat{t}_s^r + d_x^k \quad \forall \left( k,r \right) \in P_x.
\label{eq:lgtimes}
\end{equation} 

The behavior of real-time systems can be measured by the \textit{latency} $h$ and the \textit{data event distance} $\Delta \hat{t}$. The latency describes the difference between two tags in a causal event chain. The the latency is the difference between the tags $\hat{t}_s^r$ and $\hat{t}_x^k$:

\begin{equation}
h_x^k = \hat{t}_x^k - \hat{t}_s^r \quad \forall \left(k,r \right) \in P_x.
\label{eq:latency}
\end{equation} 

It is used to describe the age of information. The data event distance describes difference between the occurrences of two events at the same interface  

\begin{equation}
\Delta \hat{t}_x^k = \hat{t}_x^k - \hat{t}_x^{k-1} \quad \forall k \in \left \{1,\ldots,n \right\}.
\label{eq:dataeventdistance}
\end{equation} 

It is used as a measure for the sampling of information.

\subsection{Signal Properties}

In the following we provide formal definitions for the signal properties \textit{sampling rate}, \textit{bandwidth}, \textit{aliasing}, \textit{time delay} and \textit{synchronicity}. We define these properties for individual events and then derive a description for signals. 
We also show how these properties can be related to the known real-time measurements. 

\paragraph{Logical Data Age} The logical data age $a_{x}^k$ of an event $x^k$ is the difference between the tag $\hat{t}_{x}$ and the logical timestamp  $t_{x}$:
\begin{equation}
	a_{x}^k = \hat{t}^k_{x} - t^k_{x} 	
\label{eq:age}
\end{equation}

For the entire signal the logical data age can be described as absolute value $a_x$, if $a_x^k$ is constant. Else, it can be described using a bound of the form
$a_x^- \leq a_x^k \leq a_x^+ \forall k \in \left\{ 1,\ldots,n\right \}$.

The logical data age is similar to the latency. The difference to the latency is that the logical data age also represents algorithmic delays. The relationship between these properties can be obtained by inserting (\ref{eq:lgtimes}) and (\ref{eq:latency}) into (\ref{eq:age}). Thereby we obtain the expression:  
\begin{equation}
\begin{aligned}
a_x^k = \hat{t}_x^k-\hat{t}_s^j +d_x^k = h_x^k + d_x^k.
\end{aligned}
\label{eq:calcage}
\end{equation}

Based on this relationship, we can also determine the logical timestamp from a known latency and occurrence of an event 

\begin{equation}
t_x^k  = \hat{t}_x^k - a_x^k  = \hat{t}_x^k - h_x^k - d_x^k.
\label{eq:calclogical}
\end{equation}

This expression can be obtained by inserting (\ref{eq:calcage}) into the definition of the logical timestamp (\ref{eq:lgtimes}).

\paragraph{Data Synchronicity} The synchronicity of data $\zeta_{x_1,x_2}$ describes the difference of the logical timestamps of two values that are computed simultaneously, such that:
\begin{equation}
	\zeta_{x_1,x_2}^k = t^k_{x_1} - t^k_{x_2} \quad \forall k \in \left\{ 1,\ldots,n\right \}. 	
\label{eq:syncorg}
\end{equation}
This property is again described as a property of an event. Similar to the logical data age, we can express it for the whole signal using a bounded or an absolute expression. The data synchronicity can also be expressed as the difference of the latency of the latencies and the delays in the following form:
\begin{equation}
	\zeta_{x_1,x_2}^k = \left(h^k_{x_1} - d^k_{x_2}\right)-\left(h^k_{x_2}-d^k_{x_2}\right) \quad \forall k \in \left\{ 1,\ldots,n\right \}. 	
\label{eq:sync}
\end{equation}
This expression is obtained by inserting (\ref{eq:calclogical}) into (\ref{eq:sync}) and assuming that the tag is equal. 

\paragraph{Logical Sampling Rate} The logical sampling rate $\Delta t_{x}^k$ of an event $x^k$ measures the difference of its logical timestamp to the logical timestamp of its preceding event: 
\begin{equation}
	\Delta t_{x}^k = t_x^k-t_x^{k-1} \quad \forall k \in \left\{ 1,\ldots,n\right \}  
\label{eq:samplingrate}
\end{equation}
For the entire signal the logical data age can again be described either by an absolute value or by bounds. The logical sampling rate $\Delta t_x^k$ can be expressed as a function of the data event distance and the difference of the latencies in the following form:
\begin{equation}
\Delta t_x^k = t_x^k - t_x^{k-1} = \Delta \hat{t}_x^k - \left(h_x^k - h_x^{k-1}\right)
\label{eq:calcsr}
\end{equation}
This expression is obtained by inserting (\ref{eq:calclogical}) and (\ref{eq:dataeventdistance}) into (\ref{eq:samplingrate}), assuming that the algorithmic delay is constant. 

\paragraph{Logical Band Limit} Given a signal, whose values can be described in a spectrum. Then the logical band limit 
\begin{equation}
	l_{x} = 1/(2 f_{x}^{\max})
\label{eq:highesfreq}
\end{equation}
describes the highest frequency $f_{x}^{\max}$ in which a signal $x$ can have an amplitude that is nonzero. If there exists no spectrum (e.g. if the signal represents a discrete state), the band limit describes a lower bound on the time, in which the signal does not change its values.
As signals can generally not represent frequencies that are larger than their sampling frequency, the band limit of a signal is bounded by the logical sampling rate. An additional bound is provided by the filter operations in the components, described by $g_y$. Thus, the band limit for each pairing  $\left(\mathbf{\underline{y}}, \mathbf{\underline{u}}\right)$ and $\left(\mathbf{\underline{u}}, \mathbf{\underline{y}}\right)$ in a signal path is defined by 
\begin{equation}
		l_u  = \max \left\{l_y, \Delta t_u \right\}, \qquad 		l_{y} = \max {\left \{g_{y}, \Delta t_{y}\right\}},
	\label{eq:bandlimicomp}
\end{equation}
which means that is has to be generally determined iteratively.

\paragraph{Logical Aliasing} Given a signal which is sampled uniformly, aliasing occurs when data is undersampled. This occurs exactly when the sampling rat of a read interface is larger than the band limit of the sender. This occurs exactly when for any pair $\left(\mathbf{\underline{y}}, \mathbf{\underline{u}}\right)$ in the signal path $e_{x}$ the expression
\begin{equation}
	l_{y} \geq \Delta t_{u}
\label{eq:alias}
\end{equation}
is not true. 

\section{Relation of Signal- and Timing Requirements}

In the following we discuss the relation between signal requirements and timing requirements. This enables the transformation of specified requirements into standard formats given that the respective signal paths are known. We assume that constraints on a signal property $v_x^k$ are formulated in a bounded form 
\begin{equation}
v_x^-\leq v_x^k\leq v_x^+
\end{equation}
and that delays and filter parameters are constant. This assumption is realistic, when dealing with control systems which are often periodic.

\begin{prop}
A bounded logical data age constraint of the form $a_x^-\leq a_x^k\leq a_x^+$ will be satisfied if the condition 
\begin{equation}
a_x^- - d_x \leq h_x^k \leq a_x^+ - d_x
\end{equation}
holds for the corresponding set of causal chains $P_x$.
\end{prop}

\begin{proof}
This property can be derived by replacing $a_x^k$ with (\ref{eq:calcage}) and subtracting $d_x$. 
\end{proof}
The key aspect of this statement is that a requirement on the logical data age provides a constrained bound on the latency of the respective event-chain. 

\begin{prop}
A synchronicity constraint of the form $\zeta_{x_2,x_1}^- \leq \zeta_{x_2,x_1}^k \leq \zeta_{x_2,x_1}^+$ will be satisfied if the condition 
\begin{equation}
\zeta_{x_2,x_1}^- + d_{x_2} - d_{x_1}  \leq h_{x_2}^k - h_{x_1}^k \leq \zeta_{x_2,x_1}^+ + d_{x_2} - d_{x_1}
\end{equation}
holds for the corresponding causal chains $P_{x_1}, P_{x_2}$.
\end{prop}

\begin{proof}
We obtain this property by replacing $\zeta_{x_2,x_1}^k$ in the synchronicity constraint by (\ref{eq:calclogical}) and subtracting of the delays $d_{x_1}$ and $d_{x_1}$.  
\end{proof}
Thus, to ensure that the data is synchronous according to the constraint, it is necessary to ensure that the relative latency of the corresponding event chains stays inside of certain bounds. In AUTOSAR, this can be addressed by a constraint on the synchronicity of event chains. 

\begin{prop}
A logical sampling rate constraint of the form $\Delta t_x^- \leq \Delta t_x^k \leq \Delta t_x^+$ will be satisfied if the condition 
\begin{equation}
\Delta t_x^- \leq \left(\hat{t}_x^k - \hat{t}_x^{k-1}\right) - \left(h_{{x}}^k - h_{{x}}^{k-1} \right) \leq \Delta t_x^+
\end{equation}
holds for all events in the corresponding causal chain $P_{x}$.
\end{prop}

\begin{proof} To obtain this expression we replace the expression $\Delta t_x^k$ with (\label{calcsr}).
\end{proof}
Note, that the logical sampling rate constraint addresses a simultaneous requirement on the difference of the latencies and the difference of the tags of two consecutive events. 
 
\begin{prop}
Thus, a band limit constraint of the form $l_x^-\leq l_x^k\leq l_x^+$ can only be satisfied if the condition 
\begin{equation}
	l_x^- \geq \Delta t \geq \Delta \hat{t}_x^k - \Delta h_{x}^k 
\label{eq:blreq}
\end{equation}
is true. 
\end{prop}

\begin{proof}
In order to enable that a signal may have a specified band limit, it is necessary that the signal is sampled fast enough to represent this frequency. This is because frequencies can not be represented below the sampling rate. Therefore, the sampling rate provides a lower bound of the band limit, which can be concluded from (\ref{eq:bandlimicomp}). Hence we require $\Delta t_x<l_x^-$ is true.
\end{proof}
Note, that the sampling rate in itself can not lower the band limit. This means that an upper bound can only be enforced by the cut-off frequencies of the filters in the signal path.  


\begin{prop}
Given a no-aliasing constraint on the interface $\mathbf{\underline{x}}$ and the respective signal flow $e_{x}$. Lets assume that we can derive a subset from each signal path in $e_x$ of the form $B_x = (\mathbf{\underline{s}}, \mathbf{\underline{y}},\ldots, \mathbf{\underline{u}_{x}})$, which includes interfaces referenced to sampling and resampling behaviors and the specified interface itself. Then the no aliasing constraint will be satisfied if for any pair $(\mathbf{\underline{y}},\mathbf{\underline{u}}) \in B_x$ the condition:
\begin{equation}
	l_{y} \geq \Delta t_{u} \geq \Delta \hat{t}_{u}^k - \Delta h_{u}^k
\end{equation}
holds for all events in the respective casual chain.
\end{prop}

\begin{proof}
According to (\ref{eq:alias}) aliasing will occur if the constraint $l_{y} \leq \Delta t_{u}$ is not satisfied for any pair $\left(\mathbf{\underline{y}}, \mathbf{\underline{u}}\right)$ in a signal path. Generally the band limit of a signal can only be changed without aliasing by filtering. Also the maximum logical sampling rate can only increase along a signal path. Therefore, we only need to ensure, that the components that filter the signals read their input values with a sampling rate is not larger than the band limit of the last resampling operation. Given this requirement, the aliasing requirement can be converted into a constraint on the sampling rate for the respective event chains. 
\end{proof}
Note, that our approach requires that the band limit of the sampling interface can be determined. 

\printbibliography

\end{document}